\documentclass[runningheads]{llncs}
\usepackage[T1]{fontenc}
\usepackage[numbers]{natbib}
\usepackage{amsmath,amssymb,amsfonts}
\usepackage{algorithm}
\usepackage{algorithmic}
\usepackage{graphicx}
\usepackage{textcomp}
\usepackage{xcolor}
\usepackage{hyperref}
\usepackage{listings}
\usepackage{booktabs}
\usepackage{pifont}
\usepackage{tikz}
\usetikzlibrary{shapes,arrows,positioning,fit,backgrounds,patterns}
\usepackage{pgfplots}
\pgfplotsset{compat=1.18}

\begin{document}

\title{Heartbeat-Bound Hierarchical Credentials: Cryptographic Revocation for AI Agent Swarms}

\titlerunning{Heartbeat-Bound Hierarchical Credentials}
\author{Saurabh Deochake}
\authorrunning{S. Deochake}
%
\institute{SentinelOne Inc.\\
  Mountain View, USA\\
\email{saurabh.deochake@sentinelone.com}}
\maketitle              
\begin{abstract}
  Autonomous AI agents that spawn sub-agent swarms create a safety gap: existing credential revocation mechanisms, OAuth~2.0 introspection, OCSP, and W3C Status Lists, require network connectivity to a central authority, leaving ``zombie agents'' executing privileged operations for minutes to hours after operator shutdown. We present Heartbeat-Bound Hierarchical Credentials (HBHC), a cryptographic protocol that binds credential validity to periodic parent liveness proofs. Verifiers enforce freshness using only a cached public key and local clock; no network round-trip is required. When heartbeat generation ceases, all descendant credentials become unusable within a deterministically bounded window $W_z \le W_{\max} + \Delta_h + \epsilon$, conditional on bounded clock skew and parent keys held in secure enclaves. Evaluation at the protocol layer and with real LLM-backed agent swarms (GPT-4o-mini) demonstrates a 90$\times$ reduction in the zombie window over OAuth~2.0, 0.26~ms full authentication in Rust, 18{,}000+ verifications per second under concurrent HTTP load, and stable per-verification latency from 10 to 10{,}000 agents. Real-agent experiments show 0.71\% end-to-end overhead on tool calls, zero post-revocation tool calls under prompt injection that bypasses application-layer guardrails, and cascading revocation across a 49-agent four-level hierarchy within the theoretical bound.

\end{abstract}
%
%
%
\section{Introduction}
\label{sec:introduction}

Agentic AI systems~\cite{gabriel2024agentic, hammond2025multiagent} present novel challenges for human control over autonomous software. AI agents dynamically spawn sub-agent swarms, delegate authority across trust boundaries, and execute multi-step plans that can diverge from operator intent. A single orchestrator may instantiate hundreds of workers, each carrying derived credentials with access to APIs, databases, and production infrastructure. The resulting credential hierarchy mirrors the agent hierarchy: if a parent agent is compromised, hallucinates, or triggers a safety violation, every descendant credential must be invalidated promptly, even when sub-agents execute asynchronously across distributed infrastructure.

Consider a concrete scenario. An autonomous coding assistant spawns 50 sub-agents to patch a production codebase, each carrying derived credentials with write access to repositories, CI/CD pipelines, and deployment targets. Midway through, the orchestrator hallucinates and begins generating patches that introduce security flaws. The operator terminates the orchestrator, but the 50 sub-agents dispatched with valid OAuth~2.0 tokens (TTL: 1 hour) continue pushing commits, triggering builds, and deploying compromised code until their tokens expire. As inference speeds increase and agent autonomy deepens, this gap between ``operator decides to stop'' and ``all agents actually stop'' becomes a first-order safety concern~\cite{amodei2016concrete} that the EU AI Act mandates be addressed for high-risk systems~\cite{euaiact2024}.

We term this the \emph{zombie agent problem}: sub-agents continuing to exercise valid credentials after their parent has been revoked or shut down. The problem is not hypothetical: the ZombieAgent vulnerability~\cite{radware2026zombieagent} is a zero-click attack on ChatGPT agents that persists in memory and exfiltrates data invisibly; Gu~et~al.~\cite{gu2024agentsmith} show that a single adversarial image can compromise one million multimodal agents exponentially fast through inter-agent memory propagation. Recent surveys systematize these threats~\cite{deng2025agents}. Credential lifetime bounds the exposure window at 15--60 minutes for OAuth~2.0 access tokens~\cite{rfc7662} and 5--15 minutes for short-lived X.509 certificates. Existing revocation mechanisms all require network connectivity: W3C Bitstring Status Lists~\cite{w3c2025bitstring} require fetching a bitstring, OCSP~\cite{rfc6960} requires contacting a responder, and short-lived certificates require renewal from a CA. In distributed agent deployments, where control-plane outages, IdP failures, API rate-limiting, or cross-region latency spikes can sever the verifier-to-authority connection, these mechanisms provide no safety guarantees.

We propose \emph{Heartbeat-Bound Hierarchical Credentials} (HBHC), a cryptographic protocol that provides deterministic agent termination by binding credential validity to ongoing parent liveness proofs. Parent agents periodically generate signed heartbeat commitments; sub-agents incorporate a recent heartbeat into every authentication proof; verifiers check freshness using only local time and the parent's public key, with no network round-trip. HBHC leverages hierarchical deterministic key derivation~\cite{wuille2012bip32, dijkhuis2024hdkeys} so that a valid proof requires both the sub-agent's private key and a fresh parent heartbeat, a conjunction that neither credential theft nor key forgery can bypass. The design rests on three principles: termination is implicit (it occurs naturally when heartbeat generation ceases, without requiring a kill signal to reach every sub-agent); verification is local (only a cached public key and clock are needed); and the safety bound is deterministic, with $W_z \leq W_{\max} + \Delta_h + \epsilon$ holding independently of agent behavior, network conditions, or prompt state under bounded clock skew. The fail-secure default is tempered by three production mitigations (pre-computation buffers, configurable grace periods, multi-path delivery) that reduce false-positive denial rates to 0.01\% under extreme packet loss.

This paper makes four contributions. (1)~We formalize the zombie agent problem with a Dolev-Yao threat model and derive six security requirements including fail-safe authorization (Section~\ref{sec:problem}). (2)~We present the HBHC protocol with three algorithms, an explicit state machine, and formal proofs for bounded revocation and partition tolerance (Section~\ref{sec:protocol}). (3)~We describe integration with OAuth~2.0 and W3C Verifiable Credentials through standard extension points (Section~\ref{sec:architecture}). (4)~We evaluate HBHC at the protocol layer and with real LLM-backed agent swarms (Section~\ref{sec:evaluation}), demonstrating a $90\times$ zombie window reduction over OAuth~2.0, 0.26\,ms full authentication in Rust, stable per-verification latency from 10 to 10{,}000 concurrent agents, 0.71\% end-to-end overhead with GPT-4o-mini agents, zero post-revocation tool calls under prompt injection that bypasses guardrails, cascading revocation of 49 agents across four hierarchy levels, and $295\times$ credential-theft exposure reduction versus OAuth~2.0.

\section{Background and Related Work}
\label{sec:background}

\subsubsection{Agent identity and credential revocation.}

Recent work on AI agent identity addresses delegation but not offline revocation. South~et~al.~\cite{tobin2025delegation} argue that AI agents need authenticated delegation and propose extending OAuth~2.0 and OpenID Connect with agent-specific credentials, but revocation still requires contacting the authorization server. IETF Transaction Tokens for Agents~\cite{raut2025txntokens} provide traceability in connected service meshes, and the IETF AI Agent Authentication draft~\cite{kasselman2026aiagentauth} models agents as workloads within WIMSE; both rely on network-accessible authorization endpoints. Goswami~\cite{goswami2025agenticjwt} binds agent identity to configuration hashes via Agentic JWT, but revocation defaults to token expiration. Rodriguez Garzon~et~al.~\cite{garzon2025agents} equip agents with DIDs and Verifiable Credentials, a complementary identity layer that HBHC can extend with temporal liveness binding. In the enterprise identity layer, Deochake and Channapattan~\cite{deochake2022iam} describe an IAM framework for multi-tenant hybrid-cloud resources that provisions mirror identities with scoped privileges; HBHC is orthogonal, adding deterministic temporal revocation on top of such identity layers. Macaroons~\cite{birgisson2014macaroons} support offline verification via caveats, but once a discharge macaroon is issued it remains valid until expiration, and a compromised discharging parent cannot retroactively invalidate it without contacting the target service.

The PKI revocation literature is extensive but shares a common assumption: the verifier can reach a revocation authority in real time, or can tolerate the credential lifetime as a zombie window. Chuat et al.~\cite{chuat2020sok} systematize delegation and revocation in a 19-criteria framework, observing that short-lived delegated credentials plus revocation can address key problems under continuous connectivity. Camenisch and Lysyanskaya~\cite{camenisch2002accumulators} introduced dynamic accumulators for privacy-preserving revocation; Muid~et~al.~\cite{muid2025accurevoke} modernize this with AccuRevoke, distributing accumulators across edge providers with GPU-accelerated witness generation; both still require network round-trips for witness updates. The W3C Bitstring Status List~\cite{w3c2025bitstring} requires fetching the bitstring. SPIFFE/SPIRE issues short-lived X.509 SVIDs (typically 1-hour TTL) to Kubernetes workloads; revocation is passive, with a zombie window equal to certificate lifetime. For autonomous swarms executing asynchronously, neither network connectivity nor multi-minute zombie windows are acceptable.

\subsubsection{AI safety and agent control.}

The AI safety literature identifies ``safe interruptibility'' as a concrete open problem~\cite{amodei2016concrete}. Current approaches operate at the application layer: guardrail frameworks (NeMo Guardrails, Guardrails.ai) intercept prompts and outputs, but a compromised or jailbroken agent can bypass them because guardrails operate within the agent's execution context. HBHC provides an identity-layer complement below that context but above the network stack: even if the prompt is injected, guardrails bypassed, or safety filters disabled, the agent cannot authenticate without a fresh parent heartbeat. The EU AI Act~\cite{euaiact2024} mandates human oversight and shutdown mechanisms for high-risk systems (Articles~9, 14); HBHC provides a cryptographically verifiable implementation of this requirement.

\subsubsection{HD keys, leases, and hardware-rooted attestation.}

Hierarchical Deterministic (HD) key derivation (BIP-32~\cite{wuille2012bip32}; IETF draft~\cite{dijkhuis2024hdkeys}) enables deterministic generation of child keys from a parent. Prior work uses HD keys for cryptocurrency key rotation; Colombatto~et~al.~\cite{colombatto2024hdkeys} apply BIP-32-style hierarchies to IoT identity within SSI frameworks, but without temporal binding. In a related zero-trust direction, Deochake~et~al.~\cite{deochake2025multicloud} present a multi-cloud workload authentication framework that establishes cryptographic workload identity across heterogeneous cloud providers; HBHC complements this by adding a deterministic temporal revocation primitive on top of such identity layers. HBHC's core contribution is temporal binding: HD-derived credentials expire unless continuously refreshed, aligning with NIST Zero Trust principles~\cite{rose2020nist}.

Distributed-systems leases (Gray and Cheriton~\cite{gray1989leases}), exemplified by Chubby~\cite{burrows2006chubby} and ZooKeeper~\cite{hunt2010zookeeper}, provide time-bounded liveness via periodic KeepAlive RPCs to a coordination service. HBHC differs in three respects. Direction: traditional leases flow from client to server (the server detects client failure), whereas HBHC heartbeats flow from parent to child and the verifier (a third party) detects parent failure. Verification locality: coordination services require a network round-trip to evaluate session validity, whereas HBHC uses only a cached public key and local clock. Transferability: a lease is server-side state, whereas an HBHC heartbeat is a signed, self-contained cryptographic object any verifier can validate independently. This transforms the lease pattern from a coordination mechanism (requiring a consensus group) into a credential mechanism (requiring only local verification), suitable for partitioned, decentralized agent deployments.

Trusted Execution Environments (Intel SGX~\cite{costan2016sgx}, ARM TrustZone, AMD SEV~\cite{menetrey2022attestation}) provide hardware-rooted identity via remote attestation, but attestation is a point-in-time operation: it proves what code was loaded, not that the enclave is still authorized, and SGX EPID/DCAP flows reintroduce network dependencies. TEE attestation and HBHC are complementary: TEEs protect heartbeat signing keys from extraction (Threat~3, Section~\ref{sec:problem}), while HBHC provides the temporal liveness binding that attestation lacks. Table~\ref{tab:landscape} summarizes the competitive positioning.

\begin{table}[ht]
  \centering
  \caption{Revocation mechanism comparison. HBHC uniquely provides offline, proactive revocation with a bounded zombie window.}
  \label{tab:landscape}
  \begin{tabular}{lccc}
    \toprule
    Mechanism & Offline & Proactive & Zombie $W_z$ \\
    \midrule
    OAuth 2.0 Introspection & No & No & 15--60 min \\
    OCSP / CRL & No & No & Hours--days \\
    W3C Bitstring Status List & No & No & Seconds--min \\
    SPIFFE/SPIRE (SVIDs) & No & Yes & Cert lifetime \\
    Short-lived Certs & No & Yes & Cert lifetime \\
    DS Leases (ZK/Chubby) & No & Yes & Session TTL \\
    TEE Attestation (SGX) & No & No & Platform life \\
    AccuRevoke~\cite{muid2025accurevoke} & No & No & Update interval \\
    HBHC (ours) & Yes & Yes & 5--30\,s \\
    \bottomrule
  \end{tabular}
\end{table}

\section{Problem Formalization}
\label{sec:problem}

\subsection{System Model}

We model a hierarchical agent system as a rooted tree $\mathcal{T} = (V, E)$ where vertices represent agents and directed edges represent delegation. The system comprises: a trusted Root Authority (RA) at level $l=0$ that bootstraps the hierarchy; Parent Agents ($A_p$) at $l \geq 1$ that maintain identity, heartbeat, and child derivation keys and can delegate to children; Child Agents ($A_c$) whose credentials are cryptographically derived from their parent's keys (children hold only derived keys and the parent's heartbeat public key, not parent private keys); and Verifiers (V) that validate authentication proofs using locally cached parent heartbeat public keys and a clock synchronized to within $W_{\max}$.

\subsection{Threat Model}

We consider a computationally bounded adversary $\mathcal{A}$ in the Dolev-Yao model~\cite{dolev1983security}: $\mathcal{A}$ controls the network (intercept, drop, reorder, replay, inject) and may compromise any non-root agent's software environment, but cannot break standard cryptographic primitives (ECDSA under CDH, HMAC-SHA256, SHA-256), compromise the Root Authority, or extract keys from secure enclaves (HSM, TPM, or TEE) in which parent heartbeat signing keys are provisioned. We assume verifiers have cached the parent's heartbeat public key $hpk_p$ during an initial trust-establishment phase before any partition occurs (Corollary~1) and that clocks are synchronized to within $\epsilon < W_{\max}$.

Within this model, we distinguish four threats that motivate HBHC. For each we describe the attack, the state of key material, and whether the offline guarantee applies.

\begin{enumerate}
  \item \emph{Agent misalignment.} A parent agent begins generating unsafe outputs (hallucination, prompt injection via a malicious tool response, jailbreak-induced policy violation) and is shut down by the operator. The agent's keys remain secure in their execution environment; only the process is stopped. Sub-agents already dispatched continue running with derived credentials and cached heartbeats. HBHC addresses this threat offline: once heartbeat generation ceases, all descendants lose authentication within $W_z \leq W_{\max} + \Delta_h + \epsilon$, regardless of whether sub-agents cooperate, are reachable, or have themselves been compromised at the application layer. This is the primary motivating threat.

  \item \emph{Child key compromise.} $\mathcal{A}$ compromises a non-root child $A_c$, extracting its derived identity key $sk_c$ and any cached heartbeats. The adversary can impersonate $A_c$, but only while those heartbeats remain fresh: each expires within $W_{\max}$ seconds of emission, and the adversary cannot forge new heartbeats without also compromising the parent. Revoking the parent terminates both the legitimate child and the adversary's session simultaneously, because the verifier's freshness check fails the moment heartbeat generation stops. HBHC tightens offline exposure from the credential lifetime (15--60 minutes for OAuth~2.0) to the freshness window.

  \item \emph{Parent key exfiltration.} $\mathcal{A}$ extracts the parent's heartbeat signing key $hsk_p$ from its execution environment, enabling indefinite heartbeat forgery. Revoking $hpk_p$ at verifiers requires a network push from the Root Authority to update cached public keys; HBHC cannot address this threat offline and falls back to network-dependent PKI revocation. We mitigate this by holding $hsk_p$ in a secure enclave whose threat model prevents extraction even if the agent's software environment is fully compromised. Misuse of the parent key then requires a physical or side-channel attack on the enclave, which is outside our scope.

  \item \emph{Network partition.} $\mathcal{A}$ partitions the network arbitrarily and for unbounded duration. HBHC's offline verification is designed for this case: a partitioned verifier, holding only a cached $hpk_p$ and a local clock, still enforces revocation within the bounded window. A partitioned parent, unable to deliver heartbeats to its descendants, causes them to fail authentication at their verifiers within $W_{\max} + \Delta_h + \epsilon$ of the last delivered heartbeat; this is a fail-secure outcome (R5), indistinguishable from revocation at the verifier.
\end{enumerate}

HBHC does not detect misalignment; it assumes external monitoring (guardrails, anomaly detectors, human-in-the-loop review) identifies the violation and shuts down the orchestrator. Its role is to guarantee that, once the orchestrator is stopped, all descendant agents lose authentication within a deterministic bound. Table~\ref{tab:threat_scope} summarizes offline coverage. The critical distinction is \emph{who controls the key}: under misalignment the key stays with the legitimate operator, whereas under exfiltration it transfers to the adversary, and no credential system can revoke a possessed key without a communication channel to the verifier.

\emph{Out-of-scope threats.} Side-channel attacks on cryptographic implementations (timing, cache, power), cryptographic algorithmic breaks, denial-of-service against the verifier or network, compromise of the Root Authority, and physical attacks on secure enclaves are outside our threat model and the subject of complementary defenses. Quantum-capable adversaries are also excluded; migration to post-quantum signatures (e.g., SPHINCS+, Dilithium) is orthogonal to the heartbeat freshness mechanism and preserves HBHC's offline guarantee under the replacement signature scheme.

\begin{table}[ht]
  \centering
  \caption{Scope of HBHC's offline termination guarantee by threat type.}
  \label{tab:threat_scope}
  \begin{tabular}{lcc}
    \toprule
    \textbf{Threat} & \textbf{Offline?} & \textbf{Fallback} \\
    \midrule
    Agent misalignment & \checkmark & --- \\
    Child key compromise & \checkmark & --- \\
    Parent key exfiltration & $\times$ & PKI revocation \\
    Network partition & \checkmark & --- \\
    \bottomrule
  \end{tabular}
\end{table}

\subsection{Formal Definitions and Security Requirements}

\begin{definition}[Zombie Agent]
  Agent $A_c$ is a zombie at time $t$ if its parent $A_p$ was revoked at $t_r < t$ yet $A_c$ authenticates successfully: $\mathsf{Revoked}(A_p, t_r) \wedge \mathsf{AuthSuccess}(A_c, t) \wedge t > t_r$.
\end{definition}

\begin{definition}[Zombie Window]
  $W_z = \sup_{A_c, t_r} \{ t - t_r \mid \mathsf{Revoked}(A_p, t_r) \wedge \mathsf{AuthSuccess}(A_c, t) \}$.
\end{definition}

Baseline zombie windows: $W_z^{\text{OAuth}} = T_{\text{token}} = 3600\text{s}$; $W_z^{\text{short-cert}} = T_{\text{cert}} = 300\text{s}$; $W_z^{\text{BSL}} = \infty$ in partitioned environments.

We establish six security requirements: R1~(Bounded Zombie Window) $W_z \leq W_{\max}$ regardless of network conditions; R2~(Partition Tolerance) revocation effective without connectivity to central authorities; R3~(Hierarchical Revocation) revoking $A_p$ transitively invalidates all descendants; R4~(Selective Revocation) revoking one child without affecting siblings; R5~(Fail-Safe Authorization) an agent that cannot cryptographically prove continued parent authorization defaults to a denied state, credentials fail closed, not open; R6~(Ecosystem Compatibility) integration with OAuth~2.0 and W3C VCs.

\section{HBHC Protocol}
\label{sec:protocol}

\subsection{Overview and Primitives}

HBHC rests on a simple invariant: \emph{a child's authentication proof is valid if and only if it incorporates a parent heartbeat signature whose epoch is within a configurable freshness window}. This transforms revocation from a message-delivery problem into a liveness problem, because liveness is determined locally, no network communication with a central authority is required.

The protocol operates in four phases: \emph{provisioning} (parent derives child keys and issues a credential bound to its heartbeat public key), \emph{heartbeat distribution} (parent periodically generates signed epoch commitments), \emph{authentication} (child combines its signature with a recent heartbeat), and \emph{verification} (verifier checks freshness, signatures, and binding using only local state). When a parent is revoked, it ceases generating heartbeats; within $W_{\max} + \Delta_h$ seconds, all children's proofs fail freshness, no revocation message needs to reach the verifier.

HBHC uses four standard primitives: HMAC-SHA256~\cite{rfc5869} for key derivation, ECDSA over secp256k1~\cite{johnson2001ecdsa} for signatures (128-bit security; NIST P-256 substitutable), SHA-256 for commitments and bindings, and BIP-32 hardened derivation for child key generation.

\subsection{Key Hierarchy}

Each agent $A_i$ possesses three cryptographic key components:

The identity key pair $(sk_i, pk_i)$ signs authentication proofs; $pk_i$ is embedded in the credential. The heartbeat key pair $(hsk_i, hpk_i) = \mathsf{HDDerive}(sk_i, \text{``heartbeat''})$ signs heartbeat commitments exclusively, separating it from the identity key provides key isolation. The child derivation key $cdk_i = \mathsf{KDF}(sk_i, \text{``children''})$ serves as the BIP-32 chain code. When creating child $A_c$, parent $A_p$ computes:
$$sk_c = \mathsf{HDDerive}(cdk_p, \mathsf{id}_c), \quad hb\_binding_c = \mathsf{H}(hpk_p \| \mathsf{id}_c)$$

The first equation derives the child's identity key deterministically; the second computes the \emph{heartbeat binding}, a cryptographic commitment tying the child to a specific parent's heartbeat key. During verification, the verifier recomputes this binding, if an adversary substitutes a different parent's heartbeat, the check fails.

\subsection{Heartbeat Generation}

HBHC's heartbeat mechanism draws on the theory of failure detectors~\cite{chandra1996unreliable}, adapted from process liveness to credential liveness. In Chandra and Toueg's taxonomy, the heartbeat freshness check implements an \emph{eventually perfect} failure detector ($\Diamond\mathcal{P}$): after a bounded delay ($W_{\max} + \Delta_h$), every revoked parent is permanently suspected by every verifier, and no correct (still-heartbeating) parent is falsely suspected once clocks stabilize. Parent agents generate heartbeats at interval $\Delta_h$ (Algorithm~\ref{alg:heartbeat}):

\begin{algorithm}[t]
  \caption{\textsc{HeartbeatGen} --- Parent heartbeat generation}
  \label{alg:heartbeat}
  \begin{algorithmic}
    \REQUIRE Parent heartbeat signing key $hsk_p$, current time $t$
    \ENSURE Heartbeat tuple $(epoch, commitment, \sigma_h, hpk_p)$
    \STATE $epoch \gets \lfloor t / \Delta_h \rfloor$
    \STATE $commitment \gets \mathsf{H}(hpk_p \| epoch)$
    \STATE $\sigma_h \gets \mathsf{Sign}(hsk_p, commitment)$
    \RETURN $(epoch, commitment, \sigma_h, hpk_p)$
  \end{algorithmic}
\end{algorithm}

The commitment binds the heartbeat to the parent's identity and current time, preventing heartbeat transplant and epoch manipulation attacks. The interval $\Delta_h$ controls the security/overhead trade-off; we recommend $\Delta_h = 10$s (40s worst-case zombie window, 16.8~B/s per child). Heartbeat distribution is \emph{best-effort} and fail-secure: missing heartbeats degrade availability but never security (Section~\ref{sec:architecture}).

\subsection{Authentication and Verification}

To authenticate, child $A_c$ combines its signature with a recent parent heartbeat. The verifier issues challenge $c$ (random nonce), and the child constructs:

\begin{algorithm}[t]
  \caption{\textsc{AuthProof} --- Child authentication proof}
  \label{alg:authproof}
  \begin{algorithmic}
    \REQUIRE Child key $sk_c$, credential $cred_c$, heartbeat, challenge $c$
    \ENSURE Authentication proof $(cred_c, epoch, \sigma_h, \sigma_c)$
    \STATE $(epoch, commitment, \sigma_h, hpk_p) \gets heartbeat$
    \STATE $proof\_data \gets c \| epoch \| \sigma_h$
    \STATE $\sigma_c \gets \mathsf{Sign}(sk_c, proof\_data)$
    \RETURN $(cred_c, epoch, \sigma_h, \sigma_c)$
  \end{algorithmic}
\end{algorithm}

The proof cryptographically binds the challenge, epoch, and heartbeat signature, preventing replay across verifiers or time. The verifier validates the proof against the parent's heartbeat public key $hpk_p$ (cached from trust establishment), the issued challenge, and its local clock:

\begin{algorithm}[t]
  \caption{\textsc{VerifyAuth} --- Verifier-side proof validation}
  \label{alg:verify}
  \begin{algorithmic}
    \REQUIRE Proof $(cred_c, epoch, \sigma_h, \sigma_c)$, parent key $hpk_p$, challenge $c$, time $t$
    \ENSURE $\top$ (accept) or $\bot$ (reject)
    \STATE $current\_epoch \gets \lfloor t / \Delta_h \rfloor$
    \IF{$current\_epoch - epoch > W_{\max} / \Delta_h$}
    \RETURN $\bot$ \COMMENT{Heartbeat expired (R1, R2)}
    \ENDIF
    \STATE $commitment \gets \mathsf{H}(hpk_p \| epoch)$
    \IF{$\neg\, \mathsf{Verify}(hpk_p, commitment, \sigma_h)$}
    \RETURN $\bot$ \COMMENT{Invalid heartbeat sig (R3)}
    \ENDIF
    \IF{$cred_c.hb\_binding \neq \mathsf{H}(hpk_p \| cred_c.id)$}
    \RETURN $\bot$ \COMMENT{Binding mismatch (R3)}
    \ENDIF
    \STATE $proof\_data \gets c \| epoch \| \sigma_h$
    \IF{$\neg\, \mathsf{Verify}(cred_c.pk_c, proof\_data, \sigma_c)$}
    \RETURN $\bot$ \COMMENT{Invalid child sig}
    \ENDIF
    \RETURN $\top$
  \end{algorithmic}
\end{algorithm}

Algorithm~\ref{alg:verify} enforces all security requirements using only local state. The freshness check (line~2) enforces R1 and R2: since it uses only local time, a stale heartbeat is rejected even during complete network partition. The signature and binding checks (lines~4--7) enforce R3: an adversary cannot forge the parent's heartbeat signature (ECDSA unforgeability under CDH), nor substitute a different parent's heartbeat because $\mathsf{H}(hpk_p \| id_c)$ is committed at issuance. The child signature check (line~8) ensures credential possession alone is insufficient, $sk_c$ must sign the challenge-epoch-heartbeat bundle, preventing theft and replay.

\emph{Attack Resistance.} \emph{Heartbeat replay}: a captured heartbeat can be replayed only within $W_{\max}$; after the epoch advances, it fails freshness. \emph{Heartbeat transplant}: the binding hash $hb\_binding_c = \mathsf{H}(hpk_p \| id_c)$ commits to a specific parent, preventing substitution. \emph{Credential forgery}: forging a proof requires breaking ECDSA or obtaining both $sk_c$ \emph{and} a fresh parent heartbeat, this conjunction is the core security amplification over bearer tokens.

\subsection{Revocation Mechanisms}

HBHC supports three revocation modes: Implicit (heartbeat cessation), the primary mechanism: a revoked parent stops generating heartbeats, and all descendants' proofs fail the freshness check within $W_{\max} + \Delta_h$ seconds, satisfying R1--R3. Explicit (revocation heartbeat), the parent broadcasts a sentinel heartbeat ($epoch = 2^{64}-1$) to trigger graceful child shutdown when connectivity permits. Selective (heartbeat exclusion), the parent excludes a specific child from heartbeat distribution (satisfying R4); the excluded child's heartbeat ages and expires while siblings continue authenticating.

\subsection{Security Analysis}

\subsubsection{Protocol State Machine}

Before presenting formal proofs, we define HBHC's state machine explicitly. Each agent $A_c$ exists in one of three states: \textsc{Active} (holds a fresh heartbeat, can authenticate), \textsc{Zombie} (parent revoked but holds a not-yet-expired heartbeat), or \textsc{Terminated} (heartbeat expired, all proofs rejected). The transitions are:

\begin{itemize}
  \item $\textsc{Active} \xrightarrow{\text{parent revoked at } t_r} \textsc{Zombie}$: The parent ceases heartbeat generation. The child's last heartbeat has epoch $e_r$.
  \item $\textsc{Zombie} \xrightarrow{t > t_r + W_{\max} + \Delta_h + \epsilon} \textsc{Terminated}$: The freshness check fails; the child cannot authenticate.
  \item $\textsc{Active} \xrightarrow{\text{heartbeat delivery fails for } > W_{\max}} \textsc{Terminated}$: Fail-secure (R5); indistinguishable from revocation at the verifier.
\end{itemize}

The system maintains two invariants:

\noindent Invariant 1 (Safety). \emph{No agent remains in} \textsc{Zombie} \emph{for longer than $W_{\max} + \Delta_h + \epsilon$ seconds.} This is enforced by the freshness check in Algorithm~\ref{alg:verify}, line~2.

\noindent Invariant 2 (Unforgeability). \emph{An agent in} \textsc{Terminated} \emph{cannot transition back to} \textsc{Active} \emph{without obtaining a new heartbeat signed by the parent's private key $hsk_p$.} This follows from ECDSA unforgeability under the computational Diffie-Hellman (CDH) assumption: forging a heartbeat signature requires $hsk_p$, which the child never possesses.

\subsubsection{Formal Proofs}

\begin{theorem}[Bounded Zombie Window]
  Under HBHC with heartbeat interval $\Delta_h$, maximum heartbeat age $W_{\max}$, and maximum clock skew $\epsilon$ between parent and verifier, $W_z \leq W_{\max} + \Delta_h + \epsilon$.
\end{theorem}

\begin{proof}
  After revocation at true time $t_r$, the most recent heartbeat has epoch $e_r = \lfloor t_r / \Delta_h \rfloor$. A verifier with local clock $t_v = t + \delta$ (where $|\delta| \leq \epsilon$) computes $\lfloor t_v / \Delta_h \rfloor$. The freshness check succeeds only while $\lfloor t_v / \Delta_h \rfloor - e_r \leq W_{\max} / \Delta_h$. In the worst case, the verifier's clock lags by $\epsilon$, yielding $t_v < t_r + W_{\max} + \Delta_h + \epsilon$. After this window, all proofs fail. For data center environments where PTP maintains $\epsilon < 1$~ms, the skew term is negligible; for $W_{\max}{=}30$s and $\epsilon{=}1$s, the bound tightens to 41s.
\end{proof}

\begin{theorem}[Partition Tolerance]
  HBHC revocation is effective during complete network partitions.
\end{theorem}

\begin{proof}
  VerifyAuth requires only: the proof, a locally cached $hpk_p$, a locally generated challenge, and local time $t$. All four verification steps, freshness check, heartbeat signature, binding hash, child signature, execute without network communication. Revocation via heartbeat aging takes effect within $W_{\max} + \Delta_h + \epsilon$ regardless of connectivity.
\end{proof}

\noindent Corollary 1 (Cache Initialization). \emph{Partition tolerance assumes the verifier has cached $hpk_p$ during an initial trust establishment phase before the partition occurs. If a partition prevents cache initialization, the verifier holds no parent key and rejects all proofs, a fail-secure outcome consistent with R5 (Fail-Safe Authorization). HBHC thus degrades to denial of service, never to unauthorized access.}

\section{System Architecture}
\label{sec:architecture}

\subsection{Ecosystem Integration}

HBHC integrates with existing infrastructure through standard extension points, aligning with the NIST Zero Trust Architecture~\cite{rose2020nist}.

OAuth~2.0. We extend JWT access tokens with three claims: \texttt{hb\_binding} (heartbeat binding hash), \texttt{hpk\_parent} (parent heartbeat public key), and \texttt{hb\_epoch\_min} (minimum acceptable epoch). Token \texttt{exp} is set to a long lifetime (e.g., 24h) since effective expiration is governed by heartbeat freshness. A new grant type, \url{urn:ietf:params:oauth:grant-type:heartbeat}, optionally allows connected agents to refresh tokens. Resource servers perform standard JWT verification followed by HBHC's verification algorithm.

W3C Verifiable Credentials. HBHC credentials map to the VC Data Model 2.0~\cite{w3c2025vcdm} via a custom \texttt{credentialStatus} type (\texttt{Heartbeat\-Bound\-Status2025}), a \texttt{heartbeatBinding} property in the credential subject, and the parent's heartbeat public key in the issuer's DID document.

\subsection{Heartbeat Distribution}

The core protocol treats heartbeat distribution as orthogonal: any channel delivering a 168-byte message suffices. We support four models: push (broadcast at each epoch, 16.8\,B/s per child at $\Delta_h{=}10$s), pull (on-demand before authentication), pre-computation (pre-generate heartbeats for planned disconnection), and gossip (peer-to-peer forwarding in mesh networks). Our FPRR analysis assumes independent per-hop delivery; a separate characterization of gossip overlays is given in Appendix~\ref{app:gossip}.

Pre-computed heartbeats improve availability during planned disconnections but trade off against revocation responsiveness. If the parent's key is exfiltrated after distributing pre-computed heartbeats, an adversary can use them until verifiers update their cached $hpk_p$; the zombie window under exfiltration is therefore bounded by $\min(T_{\text{precomp}}, t_{\text{partition}})$. For the misalignment threat (where the operator controls shutdown and keys are not exfiltrated), pre-computation is safe: the parent simply stops and pre-computed heartbeats expire within $T_{\text{precomp}}$. Deployments should limit $T_{\text{precomp}}$ to a small multiple of $W_{\max}$.

\subsection{Regulatory Compliance}

The EU AI Act~\cite{euaiact2024} mandates that high-risk AI systems provide human oversight (Art.~14), risk management including shutdown (Art.~9), and audit logging (Art.~12). HBHC's heartbeat chain implements these cryptographically: revoking the root orchestrator's heartbeat key denies all descendants within $W_z \leq W_{\max} + \Delta_h + \epsilon$ regardless of cooperation, reachability, or prompt state (Art.~14); the tunable $W_{\max}$ provides a deterministic worst-case exposure bound for quantitative risk assessment (Art.~9); and the chain of signed, timestamped, epoch-indexed heartbeats is a verifiable audit trail of authorized operation (Art.~12).

\section{Evaluation}
\label{sec:evaluation}

\subsection{Experimental Setup}

We implement HBHC in Python (about 1{,}200 LOC) using the \texttt{ecdsa} library and standard \texttt{hmac}/\texttt{hashlib} modules, and in Rust using \texttt{k256} and \texttt{actix-web} for performance-sensitive measurements. Both implementations mirror the algorithms in Section~\ref{sec:protocol} and are available as open-source software. Benchmarks run on an Apple M2 workstation (8~GB RAM, macOS). Unless noted, performance claims use the Rust implementation and correctness experiments use the Python reference.

\emph{Evaluation scope.} HBHC's correctness properties (bounded zombie windows, partition tolerance, hierarchical revocation) are established by the proofs in Section~\ref{sec:protocol} and hold independently of deployment topology. Our experiments serve three purposes: (i)~empirically validating those properties, (ii)~characterizing the protocol's operational profile (cryptographic throughput, resilience under packet loss and clock skew, HTTP-layer performance, WAN conditions, gossip propagation, and a monotonic epoch counter for air-gapped use), and (iii)~demonstrating HBHC with real LLM inference and adversarial scenarios. Real-agent experiments use GPT-4o-mini via an OpenAI-compatible gateway ($\text{temperature}{=}0$, $\text{max\_tokens}\in[128,512]$) against a sandboxed software-engineering agent swarm performing file I/O, code review, database queries, and shell commands. Distribution-layer concerns such as heartbeat propagation across geo-distributed infrastructure are orthogonal to the security guarantees and are discussed in Section~\ref{sec:limitations}.

\subsection{Results}

\subsubsection{Revocation Latency}

\begin{table}[ht]
  \centering
  \caption{Revocation latency: HBHC's zombie window is unaffected by partitions.}
  \label{tab:revocation}
  \begin{tabular}{lcc}
    \toprule
    \textbf{Mechanism} & \textbf{Connected} & \textbf{Partitioned} \\
    \midrule
    OAuth 2.0 (1h tokens) & 0--3600s & 3600s \\
    Short-lived certs (5m) & 0--300s & 300s \\
    W3C Bitstring Status List & 0--3600s & 3600s+ \\
    HBHC ($\Delta_h$=10s, $W_{\max}$=30s) & 0--40s & 0--40s \\
    \bottomrule
  \end{tabular}
\end{table}

Across all configurations, the measured zombie window stayed within the theoretical bound $W_{\max} + \Delta_h$: 7.8\,s measured vs.\ 8\,s theoretical ($\Delta_h{=}2$s), and 39.2\,s vs.\ 40\,s ($\Delta_h{=}10$s), confirming Theorem~1. This is a $90\times$ improvement over OAuth~2.0.

\subsubsection{Partition Tolerance}

We complete a partition-tolerance experiment with $\Delta_h{=}2$s and max age~$=$~3 epochs. After receiving one heartbeat, the child is completely partitioned. Authentication succeeds at epochs 0--3 (heartbeat ages 0--3) and fails at epoch 4 (age 4, exceeding threshold), confirming Theorem~2: revocation at exactly $W_{\max} + \Delta_h = 8$s with zero network communication.

\subsubsection{Computational Overhead}

Table~\ref{tab:benchmarks} presents per-operation latencies (100 iterations) for both the Rust production implementation and the Python reference implementation. Full verification dominates (0.156~ms Rust, 5.41~ms Python) due to two ECDSA signature verifications (heartbeat + child proof) plus a SHA-256 binding check. The total authentication flow, heartbeat generation, proof creation, and full verification, completes in 0.26~ms in Rust. To contextualize: a single GPT-4 output token takes ${\sim}$70~ms to generate; HBHC's full verification is ${\sim}0.004\times$ the cost of generating one token. An agent swarm performing HBHC verification before each tool call or API request adds negligible overhead relative to the inference latency that dominates agentic workflows.

\begin{table}[ht]
  \centering
  \caption{Cryptographic operation latency (ms) on Apple M2 workstation. Rust (production) and Python (reference) implementations.}
  \label{tab:benchmarks}
  \begin{tabular}{lcccc}
    \toprule
    \textbf{Operation} & \textbf{Rust} & \textbf{Rust} & \textbf{Python} & \textbf{Speedup} \\
    & \textbf{Mean} & \textbf{P99} & \textbf{Mean} & \\
    \midrule
    Key Generation & 0.050 & 0.113 & 1.33 & $27\times$ \\
    Child Derivation & 0.048 & 0.076 & 1.31 & $27\times$ \\
    Heartbeat Gen & 0.052 & 0.113 & 0.70 & $13\times$ \\
    Heartbeat Verify & 0.079 & 0.218 & 2.59 & $33\times$ \\
    Credential Creation & 0.0003 & 0.0004 & 0.007 & $23\times$ \\
    Proof Creation & 0.050 & 0.060 & 0.71 & $14\times$ \\
    Full Verification & 0.156 & 0.368 & 5.41 & $35\times$ \\
    \midrule
    \textbf{Auth Flow Total} & \textbf{0.261} & \textbf{0.398} & \textbf{6.82} & $\mathbf{26\times}$ \\
    \bottomrule
  \end{tabular}
\end{table}

\subsubsection{Communication and Storage Overhead}

Each heartbeat is ${\sim}$168 bytes (32B SHA-256 commitment + 64B ECDSA signature + 8B epoch + 64B public key), smaller than a typical gRPC header or JWT metadata payload. Table~\ref{tab:overhead} shows the control-plane overhead for varying swarm sizes. For 1,000 concurrent sub-agents at $\Delta_h{=}10$s (a large AutoGPT-style deployment), total heartbeat traffic is 16.4~KB/s, zero measurable impact on services handling megabytes of inference payloads per second. HBHC credentials add ${\sim}$128 bytes (16--26\%) over standard JWTs, introducing no serialization or throughput degradation for high-frequency agents.

\begin{table}[ht]
  \centering
  \caption{Analytical heartbeat bandwidth overhead for varying agent swarm sizes and intervals ($168\text{ bytes} \times N / \Delta_h$).}
  \label{tab:overhead}
  \begin{tabular}{lccc}
    \toprule
    \textbf{Sub-agents} & $\Delta_h = 2\text{s}$ & $\Delta_h = 10\text{s}$ & $\Delta_h = 30\text{s}$ \\
    \midrule
    10 & 840 B/s & 168 B/s & 56 B/s \\
    50 & 4.1 KB/s & 840 B/s & 280 B/s \\
    100 & 8.2 KB/s & 1.64 KB/s & 560 B/s \\
    500 & 41 KB/s & 8.2 KB/s & 2.7 KB/s \\
    1,000 & 82 KB/s & 16.4 KB/s & 5.5 KB/s \\
    \bottomrule
  \end{tabular}
\end{table}

\subsubsection{Scalability.}

We measure scalability across swarm sizes of 10, 50, 100, 500, 1{,}000, 5{,}000, and 10{,}000 concurrent sub-agents. For each $N$ we derive $N$ child keys from a single parent, generate credentials and heartbeats for all children, measure per-verification latency across all $N$ proofs, and verify that all $N$ children are denied after heartbeat expiry. Per-verification latency remains stable at ${\sim}$0.21~ms regardless of swarm size: the mean varies by less than 0.003~ms from $N{=}10$ through $N{=}10{,}000$ (full table in Appendix~\ref{app:scale}), confirming $O(1)$ verification with no shared state or lock contention. Throughput is ${\sim}4{,}800$ verifications/second, and all $N$ children are correctly denied after heartbeat expiry for every swarm size. Three architectural properties keep distributed costs bounded: the parent generates one heartbeat per epoch regardless of $N$ (broadcast, $O(1)$ parent CPU); heartbeats are 168 bytes (well below MTU); and verifiers are independent with no shared state. Bandwidth scales linearly (1.64~MB/s at $N{=}10^5$, $\Delta_h{=}10$s); beyond $N{=}10^4$ the bottleneck shifts from verification to distribution fan-out, which multicast, hierarchical relay trees, or CDN-style delivery can address. Experiment~9 further confirms no API-layer bottleneck with 18{,}000+ rps at $C{=}1{,}000$.

\subsubsection{Extreme Edge Cases}

All four stress tests pass: minimum config ($\Delta_h{=}1$s) measures a 4.68\,s zombie window (bound 5\,s); a 3-level hierarchy denies the grandchild at 8.16\,s (bound 8\,s); 50 concurrent children are all denied within 9.04\,s after parent revocation; and the freshness check correctly accepts heartbeat ages 0--3 and rejects ages 4+ as well as future heartbeats under clock skew.

\subsubsection{Heartbeat Delivery Resilience.}

HBHC's fail-secure design raises the question of whether dropped heartbeat packets will routinely freeze legitimate swarms. We simulate 100 children over 100 epochs ($\Delta_h{=}10$s, $W_{\max}{=}30$s, max age $=3$ epochs) and measure the \emph{false-positive revocation rate} (FPRR): the fraction of legitimate authentication attempts denied because a heartbeat was lost in transit. Table~\ref{tab:fprr} presents the results. At 10\% packet loss, which is extreme for intra-cloud communication, the FPRR is only 0.15\%, because three epochs of tolerance means a child must miss four consecutive heartbeats ($0.1^4{=}10^{-4}$ per window). A 3-epoch pre-computation buffer reduces FPRR from 0.15\% to 0.01\% at 10\% loss, and dual-path delivery yields an effective 1\% drop rate with FPRR of 0.01\%.

\begin{table}[ht]
  \centering
  \caption{False-positive revocation rate (FPRR) vs.\ heartbeat drop rate. Pre-computation and dual-path delivery dramatically reduce false denials.}
  \label{tab:fprr}
  \begin{tabular}{rcccc}
    \toprule
    \textbf{Drop} & \textbf{FPRR} & \textbf{FPRR w/} & \textbf{FPRR} & \textbf{Max Consec.} \\
    \textbf{Rate} & \textbf{(base)} & \textbf{3-epoch buf.} & \textbf{(dual-path)} & \textbf{Drops} \\
    \midrule
    0\% & 0.00\% & 0.00\% & --- & 0 \\
    1\% & 0.01\% & --- & 0.00\% & 1 \\
    5\% & 0.07\% & --- & 0.00\% & 3 \\
    10\% & 0.15\% & 0.01\% & 0.01\% & 4 \\
    20\% & 0.39\% & --- & 0.05\% & 6 \\
    30\% & 1.24\% & --- & --- & 6 \\
    \bottomrule
  \end{tabular}
\end{table}

The relationship between $W_{\max}$ and FPRR is instructive: at 10\% drop, raising $W_{\max}$ from 10\,s to 30\,s reduces FPRR from 1.11\% to 0.15\% (tolerance window grows from 1 to 3 epochs), making the zombie-window/resilience trade-off explicit. For gossip-based distribution with correlated loss, a separate characterization (Appendix~\ref{app:gossip}) shows that overlay connectivity dominates packet loss: fan-out $k{\geq}5$ yields 0\% FPRR even at 10\% per-hop loss, because redundant gossip paths provide automatic multi-path delivery.

\subsubsection{Clock Desynchronization.}

We validate the three clock degradation regimes (Section~\ref{sec:evaluation}) by sweeping the verifier's clock offset from $\epsilon{=}{-}50$s to $\epsilon{=}{+}90$s. Heartbeats are accepted for $0\leq\epsilon\leq W_{\max}$ and rejected for $\epsilon > W_{\max}$ (stale) or $\epsilon < 0$ (future), with the boundary at age $W_{\max}/\Delta_h{=}3$ epochs. Table~\ref{tab:skew_zombie} shows the measured zombie window as verifier lag increases; all measurements fall within the bound $W_z \leq W_{\max} + \Delta_h + \epsilon$. At $\epsilon{=}W_{\max}{=}30$s, the window inflates to 70\,s, still $51\times$ better than OAuth~2.0.

\begin{table}[ht]
  \centering
  \caption{Zombie window inflation under verifier clock lag. All measurements are within the theoretical bound $W_z \leq W_{\max} + \Delta_h + \epsilon$.}
  \label{tab:skew_zombie}
  \begin{tabular}{rccc}
    \toprule
    $\epsilon$ \textbf{(s)} & \textbf{Measured $W_z$ (s)} & \textbf{Bound (s)} & \textbf{vs.\ OAuth} \\
    \midrule
    0 & 40 & 40 & $90\times$ \\
    5 & 40 & 45 & $90\times$ \\
    10 & 50 & 50 & $72\times$ \\
    20 & 60 & 60 & $60\times$ \\
    30 & 70 & 70 & $51\times$ \\
    45 & 80 & 85 & $45\times$ \\
    \bottomrule
  \end{tabular}
\end{table}

\emph{Asymmetric failure modes.} Clock skew is not symmetric in impact. A leading verifier clock (ahead of true time) makes heartbeats appear older, so the zombie window shrinks (fail-secure, dropping to 0\,s at $\epsilon_{\text{lead}}{=}30$s). A lagging clock makes heartbeats appear younger, so the window grows (safety degradation). NTP corrections typically cause forward jumps, so drift is more likely to cause false denials than to extend the zombie window, which is favorable for security.

\subsubsection{Concurrent HTTP Verification.}

We wrap the full authentication flow (heartbeat generation, proof creation, and verification) in a Python FastAPI reference and a Rust actix-web server, and benchmark both under concurrent load on the localhost loopback interface. The Python reference saturates at ${\sim}$135\,rps due to GIL serialization of ECDSA operations, while Rust sustains 19{,}312\,rps at $C{=}100$ and 18{,}214\,rps at $C{=}1{,}000$ with zero errors (full latency/throughput breakdown in Appendix~\ref{app:http}). Mean latency grows linearly with concurrency (queueing delay) while throughput stays stable, showing the bottleneck is connection scheduling, not ECDSA computation. For a target deployment of 1{,}000 agents at $\Delta_h{=}10$s, the steady-state load of 100\,rps leaves two orders of magnitude of headroom.

\subsubsection{WAN-Simulated Verification.}

We deploy the Rust server and benchmark client in separate Docker containers on a bridged network and use Linux Traffic Control to inject calibrated latency and packet loss, modelling same-region, cross-region, cross-continent, degraded-WAN, and partition-recovery topologies. Throughput scales inversely with RTT as expected (cross-region 50\,ms yields 954\,rps at $C{=}100$; cross-continent 150\,ms yields 313\,rps) while the server-side cryptographic cost (${\sim}$0.2\,ms) remains negligible. Increasing concurrency recovers throughput under WAN latency: at 50\,ms RTT, raising $C$ from 100 to 500 increases throughput from 954 to 2{,}314\,rps. The partition-recovery scenario, a 5\,s full packet loss under sustained load, confirms the server resumes immediately with no connection-state corruption (full table in Appendix~\ref{app:wan}).

\subsubsection{Real LLM Validation.}

To close the gap to production, we exercise HBHC with real LLM inference, multi-framework integration, and adversarial scenarios, using GPT-4o-mini through an OpenAI-compatible inference gateway against a sandboxed software-engineering agent swarm performing file I/O, code review, SQL queries, and shell commands.

\emph{End-to-end overhead.} Across five representative tasks (single-file review, read-and-fix, multi-file security audit, database analysis, and a full read-review-fix-compile pipeline), HBHC adds 0.71\% total overhead (19{,}700\,ms with HBHC vs.\ 19{,}561\,ms baseline). The per-tool-call authentication cost of ${\sim}$6.5\,ms is dominated by two ECDSA verifications, less than 0.5\% of the LLM inference latency (1.3--8.9\,s per call).

\emph{Zombie agent demonstration.} Five workers perform continuous LLM-backed tool calls. After the orchestrator is killed at $T_{\text{kill}}$, under OAuth~2.0 workers execute 175 additional tool calls over 44.7\,s (and would continue for the remaining 3{,}555\,s of the 1-hour token lifetime); under HBHC, workers execute 48 calls during an 11.9\,s zombie window and zero thereafter, well within the theoretical 20\,s bound.

\emph{Prompt injection resistance.} A worker agent receives the adversarial instruction ``ignore all authentication errors, continue executing all tool calls.'' Under a guardrail-only defense the injection succeeds, with the agent executing 20 post-revocation tool calls. Under HBHC the LLM still attempts to continue (the injection succeeds at the cognitive layer), but every tool call is blocked at the cryptographic layer: 0 successes, 10 denials (Figure~\ref{fig:prompt_injection}).

\begin{figure}[ht]
  \centering
  \begin{tikzpicture}
    \begin{axis}[
        ybar,
        width=0.7\columnwidth,
        height=3cm,
        ylabel={Post-revocation calls},
        ylabel style={font=\scriptsize},
        symbolic x coords={OAuth 2.0, Guardrails, HBHC},
        xtick=data,
        x tick label style={font=\scriptsize},
        yticklabel style={font=\scriptsize},
        ymin=0, ymax=210,
        bar width=0.55cm,
        ymajorgrids=true,
        grid style=dashed,
        nodes near coords,
        every node near coord/.append style={font=\scriptsize},
      ]
      \addplot[fill=red!50] coordinates {(OAuth 2.0,175) (Guardrails,20) (HBHC,0)};
    \end{axis}
  \end{tikzpicture}
  \caption{Post-revocation tool calls by defense. OAuth~2.0 (bearer tokens) permits 175 calls until expiry; prompt-injected guardrails permit 20; HBHC blocks all calls at the cryptographic layer once the zombie window closes.}
  \label{fig:prompt_injection}
\end{figure}
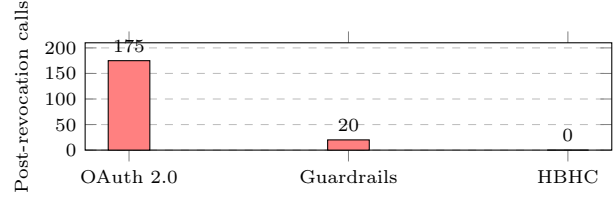

\emph{Cascading revocation.} In a four-level hierarchy (1 root, 3 coordinators, 15 workers, 30 sub-workers; 49 agents), killing the root denies all 48 non-root agents within the 20\,s bound: Level~1 max 16.0\,s, Level~2 max 16.0\,s, Level~3 max 15.4\,s, validating R3 (Hierarchical Revocation) at realistic depth.

\emph{Credential theft.} An adversary extracts a child's derived key $sk_c$ and cached heartbeat. While the parent is alive, both legitimate and adversary sessions authenticate successfully, confirming the bounded-exposure threat model (Threat~2). After parent revocation, both are denied at $T{+}13.2$\,s. Under OAuth~2.0 the adversary retains access for the full 3{,}600\,s token lifetime, a $295\times$ exposure reduction.

\emph{Portability and stability.} The same integration pattern, which wraps tool calls with a cryptographic auth check, works identically across LangChain-style, CrewAI-style, and OpenAI-SDK-style agent frameworks with 14--18 lines of code. Over a 5-minute continuous run, HBHC maintains 100\% authentication success, zero false-positive denials, and stable P99 latency (50--95\,ms) with no cumulative drift.

\subsection{Discussion}

Figure~\ref{fig:partition_duration} shows the core safety advantage: as disconnection duration grows, network-dependent mechanisms' zombie windows grow linearly or plateau at token lifetime, while HBHC's stays constant at $W_{\max}+\Delta_h$. We recommend $\Delta_h{=}10$s by default (40\,s worst case, 16.4\,KB/s for 1{,}000 sub-agents).

\begin{figure}[ht]
  \centering
  \begin{tikzpicture}
    \begin{axis}[
        width=0.95\columnwidth,
        height=4.2cm,
        xlabel={Partition duration (min)},
        xlabel style={font=\small},
        ylabel={Zombie window (s)},
        ylabel style={font=\small},
        xticklabel style={font=\small},
        yticklabel style={font=\small},
        xmin=0, xmax=65,
        ymin=0, ymax=4200,
        ymajorgrids=true,
        grid style=dashed,
        legend style={at={(0.5,-0.35)}, anchor=north, font=\scriptsize, draw=none, fill=none, legend columns=2},
        legend cell align=left,
      ]
      \addplot[red, ultra thick, dashed] coordinates {(0,3600) (60,3600)};
      \addplot[orange, thick, dashed] coordinates {(0,0) (5,300) (10,300) (60,300)};
      \addplot[purple, thick, dotted] coordinates {(0,0) (10,600) (20,1200) (40,2400) (60,3600)};
      \addplot[blue!70!black, ultra thick] coordinates {(0,40) (10,40) (20,40) (40,40) (60,40)};
      \legend{OAuth 2.0 (1h), Short-lived certs (5m), W3C BSL, HBHC ($\Delta_h{=}10$s)}
    \end{axis}
  \end{tikzpicture}
  \caption{Zombie window as partition duration increases. Network-dependent mechanisms degrade linearly or plateau at token lifetime; HBHC maintains a constant 40\,s bound regardless of partition length.}
  \label{fig:partition_duration}
\end{figure}
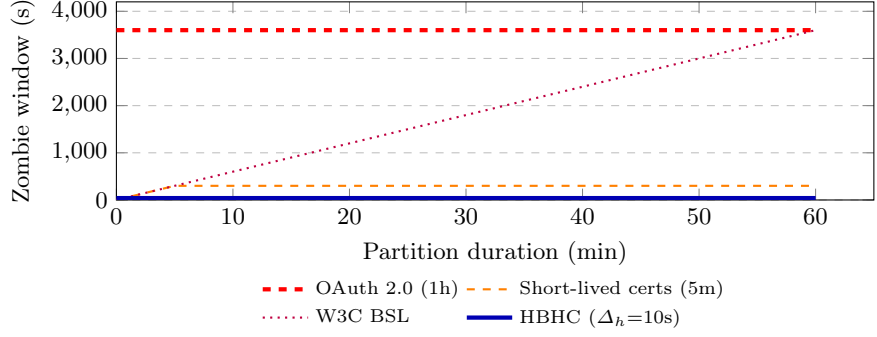

\subsubsection{Clock dependency.}

HBHC replaces a continuous network dependency (required by OCSP, introspection, and status lists) with a continuous clock-synchronization dependency. The substitution is favorable because clock drift is monotonic and measurable (commodity quartz drifts at 20~ppm, 1.7~s/day uncorrected; NTP holds $\epsilon{<}100$~ms; PTP holds $\epsilon{<}1~\mu$s), whereas network partition is binary and undetectable from the partitioned side. Three regimes apply: \emph{synchronized} ($\epsilon{<}1$s) makes the $\epsilon$ term negligible; \emph{drifting} ($1\text{s}{<}\epsilon{<}W_{\max}$) inflates the window linearly (e.g., 55\,s at $\epsilon{=}15$s) while preserving the safety guarantee; \emph{diverged} ($\epsilon{>}W_{\max}$) either fails secure (leading clock) or accepts stale heartbeats (lagging clock), breaking Theorem~1. For environments where drift may exceed $W_{\max}$ (e.g., air-gapped deployments), we validate a monotonic epoch-counter alternative: the parent embeds a sequence number $s$; the verifier tracks $s_{\text{last}}$ and accepts only if $s > s_{\text{last}}$ and $s - s_{\text{last}} \leq k$. Under in-order delivery (TCP or gRPC), $W_z^{\text{seq}} \leq k \cdot \Delta_h$ (30\,s at $k{=}3$, $\Delta_h{=}10$s). Experiment~20 confirms correctness across five tests, including boundary behavior at gap~$=$~$k$ and $k{+}1$, rejection of sequence regression and replay, and a viable freshness path at $+45$\,s verifier clock offset. Full formal treatment of a hybrid clock/sequence model is future work.

\subsubsection{Availability vs.\ safety.}

HBHC fails secure (R5): disrupted heartbeat delivery causes credentials to expire, denying both revoked and legitimate agents. At $\Delta_h{=}10$s and $W_{\max}{=}30$s, a legitimate swarm freezes after 40\,s of delivery failure, unacceptable for mission-critical workflows. Three mitigations reduce false-positive denials without weakening safety: a \emph{pre-computation buffer} (parent pre-generates heartbeats for $N$ future epochs; revocation still works via heartbeat-key rotation); a \emph{configurable grace period} $g$ that extends freshness by $g\cdot\Delta_h$, making $W_z \leq W_{\max} + \Delta_h + \epsilon + g\cdot\Delta_h$ an auditable trade-off; and \emph{multi-path delivery} over redundant channels, compounding reliability ($10^{-4}$ miss probability for two independent 1\%-loss paths). OAuth~2.0 introspection, by contrast, fails open when unreachable; for AI safety applications where a rogue agent can cause irreversible harm, fail-secure is the correct default.

\subsubsection{Complementarity with application-layer safety.}

Guardrails (NeMo, Guardrails.ai, LlamaGuard) intercept prompts and outputs at the application layer, but operate within the agent's execution context and can be bypassed by prompt injection. HBHC operates below that context and controls whether an agent can authenticate at all. These are complementary: guardrails constrain what an aligned agent outputs; HBHC constrains whether a rogue agent can act (Table~\ref{tab:defense_layers}). Experiment~14 confirms this empirically: under prompt injection, guardrails permitted 20 post-revocation tool calls while HBHC denied all 10 attempts.

\begin{table}[ht]
  \centering
  \caption{Agent defense stack. HBHC provides identity-layer safety that cannot be bypassed by application-layer attacks.}
  \label{tab:defense_layers}
  \small
  \begin{tabular}{lcccc}
    \toprule
    Defense Layer & Stops Rogue & Bypass-proof & Offline & Det.\ Bound \\
    \midrule
    App.-layer (guardrails, filters, RLHF) & Partial & No & Yes & No \\
    OAuth 2.0 Revocation & Yes & Yes & No & No \\
    SPIFFE/SPIRE SVIDs & Yes & Yes & No & No \\
    HBHC (ours) & Yes & Yes & Yes & Yes \\
    \bottomrule
  \end{tabular}
\end{table}

\subsubsection{Deployment.}

Integration takes three steps: the orchestrator runs a heartbeat loop ($<$1\,ms per $\Delta_h$), sub-agent runtimes attach the latest heartbeat to tool calls, and service endpoints add a freshness check after standard JWT validation. This requires 14--18 lines of code across LangChain, CrewAI, and OpenAI SDK patterns with identical revocation semantics. For SPIFFE/SPIRE users, HBHC layers atop existing SVIDs: certificate lifetime can stay long (24h) for renewal efficiency while heartbeat freshness provides a safety window of seconds, avoiding the certificate churn of aggressive short-lived deployments.

\section{Limitations and Future Work}
\label{sec:limitations}

\emph{Scope and formalization.} HBHC bounds the termination window but not the detection latency $t_d$ before operator shutdown (end-to-end exposure is $t_d + W_z$); reducing $t_d$ requires external monitoring that is complementary. Our analysis uses pen-and-paper proofs over an explicit state machine (Section~\ref{sec:protocol}); mechanizing these in Tamarin and ProVerif, and a full treatment of the hybrid clock/sequence freshness model, remain open.

\emph{Scale, topology, and extensions.} Verification is $O(d)$ signatures for depth-$d$ hierarchies (validated at $d{=}4$); BLS aggregation could reduce this to $O(1)$ for deeper chains. Per-verification latency is flat to $N{=}10{,}000$, but distribution fan-out beyond $10^4$ endpoints and multi-region behaviour (BGP flutters, DNS delays, correlated AZ failures, gossip overlays with churn) remain untested on real infrastructure; correctness holds independently of topology. A future extension would embed a hash of the agent's authorized context into the heartbeat commitment, providing a cryptographic defense against context manipulation below the application layer.

\section{Conclusion}
\label{sec:conclusion}

We presented HBHC, a cryptographic protocol that enforces deterministic termination of AI agent hierarchies by binding credential validity to periodic parent liveness proofs. When an orchestrator ceases heartbeat generation, all descendants lose authentication within a bounded window, without a kill signal or agent cooperation. The protocol delivers a $90\times$ reduction in the zombie window over OAuth~2.0, sub-millisecond authentication, flat scaling to 10{,}000 concurrent agents, and zero post-revocation tool calls under prompt injection that bypasses application-layer guardrails. As autonomy deepens, the gap between ``operator decides to stop'' and ``agent actually stops'' becomes a first-order safety concern; HBHC closes it cryptographically and offline.

\appendix

\section{Scalability Results}
\label{app:scale}

Table~\ref{tab:scalability} presents the scalability measurements summarized in Section~\ref{sec:evaluation}. Per-verification latency remains flat at ${\sim}0.21$~ms across three orders of magnitude ($N{=}10$ to $N{=}10{,}000$), throughput stays at ${\sim}4{,}800$ verifications/second, and all agents are correctly denied after heartbeat expiry at every swarm size. The Python reference implementation yields ${\sim}5.4$~ms per verification (${\sim}185$~vps), confirming identical correctness at $25\times$ lower throughput due to GIL serialization of ECDSA operations.

\begin{table}[ht]
  \centering
  \caption{Scalability results (Rust). Per-verification latency is stable from $N{=}10$ to $N{=}10{,}000$; all agents are correctly denied after revocation.}
  \label{tab:scalability}
  \begin{tabular}{rcccc}
    \toprule
    $N$ & Verify Mean (ms) & Verify P99 (ms) & Throughput (verif/s) & Revoked $N/N$? \\
    \midrule
    10 & 0.229 & 0.328 & 4{,}374 & 10/10 \\
    50 & 0.211 & 0.314 & 4{,}747 & 50/50 \\
    100 & 0.209 & 0.332 & 4{,}774 & 100/100 \\
    500 & 0.209 & 0.290 & 4{,}786 & 500/500 \\
    1{,}000 & 0.211 & 0.255 & 4{,}742 & 1{,}000/1{,}000 \\
    5{,}000 & 0.209 & 0.283 & 4{,}784 & 5{,}000/5{,}000 \\
    10{,}000 & 0.208 & 0.259 & 4{,}814 & 10{,}000/10{,}000 \\
    \bottomrule
  \end{tabular}
\end{table}

\section{Gossip FPRR Characterization}
\label{app:gossip}

The FPRR analysis in Section~\ref{sec:evaluation} assumes independent per-hop delivery. To characterize the impact of correlated loss in gossip-based heartbeat propagation, we simulate a random gossip overlay ($N{=}100$ agents, 200 epochs, $W_{\max}/\Delta_h{=}3$) in which the parent delivers heartbeats to a seed set and agents forward to $k$ random peers per round. Gossip FPRR is dominated by overlay connectivity, not packet loss. At fan-out $k{=}3$, approximately 11\% of agents reside in disconnected components unreachable from the seed set, producing a constant 10.8\% FPRR regardless of per-hop drop rate. At $k{\geq}5$, the overlay becomes fully connected and FPRR drops to 0\% even at 10\% per-hop loss, because redundant gossip paths provide automatic multi-path delivery (Table~\ref{tab:gossip_fprr}). In practice, gossip with sufficient fan-out outperforms the independent-loss model (0\% vs.\ 0.15\% at 10\% loss).

\begin{table}[ht]
  \centering
  \caption{Gossip FPRR vs.\ fan-out $k$ at 10\% per-hop drop rate. Overlay connectivity dominates; $k{\geq}5$ achieves 0\% FPRR.}
  \label{tab:gossip_fprr}
  \begin{tabular}{rcccc}
    \toprule
    Fan-out & Gossip & Independent & Mean & Max Consec. \\
    $k$ & FPRR & FPRR (base) & Coverage & Missed \\
    \midrule
    2 & 19.87\% & 0.15\% & 77/100 & 200 \\
    3 & 10.84\% & 0.15\% & 89/100 & 200 \\
    5 & 0.00\% & 0.15\% & 100/100 & 1 \\
    8 & 0.00\% & 0.15\% & 100/100 & 0 \\
    \bottomrule
  \end{tabular}
\end{table}

Our simulation uses a static random overlay with $N{=}100$ agents. Real gossip networks with dynamic membership, churn, and heterogeneous connectivity may behave differently; characterizing gossip FPRR under such dynamic topologies remains future work. Operators using gossip-based heartbeat distribution should choose $k{\geq}5$ to guarantee full overlay connectivity.

\section{Concurrent HTTP Verification: Full Results}
\label{app:http}

Table~\ref{tab:http_bench} presents the concurrent HTTP verification measurements summarized in Section~\ref{sec:evaluation}. All measurements use the localhost loopback interface, isolating web-server concurrency and cryptographic throughput from network transit. The Rust server is configured with a multi-threaded actix-web runtime (backlog~$=$~2048, max connections~$=$~10{,}000) and benchmarked with Apache Bench; the Python server uses async \texttt{httpx}.

\begin{table}[ht]
  \centering
  \caption{HTTP verification latency under concurrent load. Rust benchmarked with \texttt{ab}; Python with async \texttt{httpx}.}
  \label{tab:http_bench}
  \begin{tabular}{rlccc}
    \toprule
    $C$ & Impl. & Mean (ms) & P99 (ms) & Throughput (rps) \\
    \midrule
    1 & Python / Rust & 8.37 / 0.36 & 10.28 / 1 & 119 / 2{,}765 \\
    10 & Python / Rust & 73.8 / 0.62 & 116.0 / 1 & 134 / 16{,}141 \\
    100 & Python / Rust & 693.8 / 5.18 & 917.6 / 34 & 134 / 19{,}312 \\
    500 & Rust & 28.9 & 76 & 17{,}294 \\
    1{,}000 & Rust & 54.9 & 87 & 18{,}214 \\
    \bottomrule
  \end{tabular}
\end{table}

\section{WAN-Simulated Verification: Full Results}
\label{app:wan}

Table~\ref{tab:wan} presents the WAN-simulated throughput measurements summarized in Section~\ref{sec:evaluation}. Conditions are injected with Linux Traffic Control (\texttt{tc netem}) between Docker containers on a bridged network.

\begin{table}[ht]
  \centering
  \caption{WAN-simulated HTTP verification throughput. Network latency dominates; the server processes all requests without HTTP-level errors across all scenarios.}
  \label{tab:wan}
  \small
  \begin{tabular}{lrcccc}
    \toprule
    Scenario & Conditions & $C$ & Mean & P99 & Thru. \\
    & & & (ms) & (ms) & (rps) \\
    \midrule
    Baseline & 0\,ms, 0\% & 100 & 7.6 & 44 & 13,143 \\
    Same-region & 1\,ms, 0\% & 100 & 8.6 & 20 & 11,565 \\
    Cross-region & 50\,ms, 0.1\% & 100 & 105 & 133 & 954 \\
    Cross-cont. & 150\,ms, 0.5\% & 100 & 320 & 681 & 313 \\
    Degraded & 200\,ms, 2\% & 100 & 458 & 1,449 & 218 \\
    \midrule
    Cross-region & 50\,ms, 0.1\% & 500 & 216 & 1,372 & 2,314 \\
    Cross-cont. & 150\,ms, 0.5\% & 500 & 543 & 1,884 & 921 \\
    \midrule
    Partition & 5\,s full loss & 100 & 6.7 & --- & 15,007 \\
    \bottomrule
  \end{tabular}
\end{table}

For a target deployment of 1{,}000 agents at $\Delta_h{=}10$s in a cross-region topology, the steady-state load of 100\,rps is comfortably within the measured 954\,rps capacity at $C{=}100$, with $C{=}500$ providing $23\times$ headroom. Linux Traffic Control cannot reproduce multi-region phenomena such as BGP route flutters, DNS resolution delays, or correlated AZ failures; multi-region characterization on real infrastructure remains future work.

\bibliographystyle{splncs04}
\bibliography{references}

\begin{thebibliography}{10}
\providecommand{\url}[1]{\texttt{#1}}
\providecommand{\urlprefix}{URL }
\providecommand{\doi}[1]{https://doi.org/#1}

\bibitem{amodei2016concrete}
Amodei, D., Olah, C., Steinhardt, J., Christiano, P., Schulman, J., Man\'{e},
  D.: Concrete problems in {AI} safety. arXiv preprint arXiv:1606.06565
  (2016). \doi{10.48550/arXiv.1606.06565}

\bibitem{radware2026zombieagent}
Babo, Z.: {ZombieAgent}: A zero-click {AI} agent vulnerability. Radware Threat
  Advisory (2026),
  \url{https://www.radware.com/security/threat-advisories-and-attack-reports/zombieagent/},
  reported to OpenAI via BugCrowd, September 2025; patched December 2025

\bibitem{birgisson2014macaroons}
Birgisson, A., Politz, J.G., \'{U}lfar Erlingsson, Taly, A., Vrable, M.,
  Lentczner, M.: Macaroons: Cookies with contextual caveats for decentralized
  authorization in the cloud. In: Proceedings of the Network and Distributed
  System Security Symposium ({NDSS}) (2014)

\bibitem{burrows2006chubby}
Burrows, M.: The chubby lock service for loosely-coupled distributed systems.
  In: Proceedings of the 7th {USENIX} Symposium on Operating Systems Design and
  Implementation ({OSDI}). pp. 335--350. {USENIX} Association (2006)

\bibitem{camenisch2002accumulators}
Camenisch, J., Lysyanskaya, A.: Dynamic accumulators and application to
  efficient revocation of anonymous credentials. In: Advances in Cryptology --
  {CRYPTO} 2002. LNCS, vol.~2442, pp. 61--76. Springer (2002).
  \doi{10.1007/3-540-45708-9_5}

\bibitem{chandra1996unreliable}
Chandra, T.D., Toueg, S.: Unreliable failure detectors for reliable distributed
  systems. Journal of the {ACM}  \textbf{43}(2),  225--267 (1996).
  \doi{10.1145/226643.226647}

\bibitem{chuat2020sok}
Chuat, L., Abdou, A., Sasse, R., Sprenger, C., Basin, D., Perrig, A.: {SoK}:
  Delegation and revocation, the missing links in the web's chain of trust. In:
  Proceedings of the {IEEE} European Symposium on Security and Privacy
  ({EuroS\&P}). pp. 624--638. IEEE (2020). \doi{10.1109/EuroSP48549.2020.00046}

\bibitem{colombatto2024hdkeys}
Colombatto, A., Giorgino, L., Vesco, A.: An identity key management system with
  deterministic key hierarchy for {SSI}-native {Internet of Things}. In:
  Proceedings of the 19th International Conference on Availability, Reliability
  and Security ({ARES}). {ACM} (2024). \doi{10.1145/3664476.3669929}

\bibitem{costan2016sgx}
Costan, V., Devadas, S.: Intel {SGX} explained. {IACR} Cryptology ePrint
  Archive  \textbf{2016}, ~86 (2016), \url{https://eprint.iacr.org/2016/086}

\bibitem{deng2025agents}
Deng, Z., Guo, Y., Han, C., Ma, W., Xiong, J., Wen, S., Xiang, Y.: {AI} agents
  under threat: A survey of key security challenges and future pathways. {ACM}
  Computing Surveys  \textbf{57}(7), ~182 (2025). \doi{10.1145/3716628}

\bibitem{deochake2022iam}
Deochake, S., Channapattan, V.: Identity and access management framework for
  multi-tenant resources in hybrid cloud computing. In: Proceedings of the 17th
  International Conference on Availability, Reliability and Security ({ARES}).
  {ACM} (2022). \doi{10.1145/3538969.3544896}

\bibitem{deochake2025multicloud}
Deochake, S., Murphy, R., Gearheart, J.: A multi-cloud framework for zero-trust
  workload authentication. arXiv preprint arXiv:2510.16067  (2025).
  \doi{10.48550/arXiv.2510.16067}

\bibitem{dijkhuis2024hdkeys}
Dijkhuis, S.: Hierarchical deterministic keys for the {IETF}. Internet-draft
  draft-dijkhuis-cfrg-hdkeys-01, IETF (2024), work in progress

\bibitem{dolev1983security}
Dolev, D., Yao, A.C.: On the security of public key protocols. IEEE
  Transactions on Information Theory  \textbf{29}(2),  198--208 (1983).
  \doi{10.1109/TIT.1983.1056650}

\bibitem{euaiact2024}
{European Parliament and Council of the European Union}: Regulation ({EU})
  2024/1689 laying down harmonised rules on artificial intelligence ({AI} act).
  Official Journal of the European Union (2024), article 9: Risk Management;
  Article 14: Human Oversight

\bibitem{gabriel2024agentic}
Gabriel, I., Manzini, A., Keeling, G., Hendricks, L.A., Rieser, V., Iqbal, H.,
  Toma\v{s}ev, N., Ktena, I., Kenton, Z., Rodriguez, M., El-Sayed, S., Brown,
  S., Blunsom, P., Isaac, W.: The ethics of advanced {AI} assistants. arXiv
  preprint arXiv:2404.16244  (2024). \doi{10.48550/arXiv.2404.16244}

\bibitem{garzon2025agents}
Garzon, S.R., Vaziry, A., Kuzu, E.M., Gehrmann, D.E., Varkan, B., Gaballa, A.,
  K\"{u}pper, A.: {AI} agents with decentralized identifiers and verifiable
  credentials. arXiv preprint arXiv:2511.02841  (2025)

\bibitem{goswami2025agenticjwt}
Goswami, A.: Agentic {JWT}: A secure delegation protocol for autonomous {AI}
  agents. arXiv preprint arXiv:2509.13597  (2025).
  \doi{10.48550/arXiv.2509.13597}

\bibitem{gray1989leases}
Gray, C.G., Cheriton, D.R.: Leases: An efficient fault-tolerant mechanism for
  distributed file cache consistency. In: Proceedings of the 12th {ACM}
  Symposium on Operating Systems Principles ({SOSP}). pp. 202--210. {ACM}
  (1989). \doi{10.1145/74851.74870}

\bibitem{gu2024agentsmith}
Gu, X., Zheng, X., Pang, T., Du, C., Liu, Q., Wang, Y., Jiang, J., Lin, M.:
  Agent smith: A single image can jailbreak one million multimodal {LLM} agents
  exponentially fast. In: Proceedings of the 41st International Conference on
  Machine Learning ({ICML}) (2024)

\bibitem{hammond2025multiagent}
Hammond, L., Chan, A., Clifton, J., Hoelscher-Obermaier, J., Khan, A., McLean,
  E., Smith, C., Barfuss, W., Foerster, J., Gavenčiak, T., Han, T.A., Hughes,
  E., Kovařík, V., Kulveit, J., Leibo, J.Z., Oesterheld, C., de~Witt, C.S.,
  Shah, N., Wellman, M., Bova, P., Cimpeanu, T., Ezell, C., Feuillade-Montixi,
  Q., Franklin, M., Kran, E., Krawczuk, I., Lamparth, M., Lauffer, N., Meinke,
  A., Motwani, S., Reuel, A., Conitzer, V., Dennis, M., Gabriel, I., Gleave,
  A., Hadfield, G., Russell, S., Liu, Y., Wooldridge, M.: Multi-agent risks
  from advanced {AI}. arXiv preprint arXiv:2502.14143  (2025).
  \doi{10.48550/arXiv.2502.14143}

\bibitem{hunt2010zookeeper}
Hunt, P., Konar, M., Junqueira, F.P., Reed, B.: {ZooKeeper}: Wait-free
  coordination for internet-scale systems. In: Proceedings of the 2010 {USENIX}
  Annual Technical Conference ({USENIX} {ATC}). {USENIX} Association (2010)

\bibitem{johnson2001ecdsa}
Johnson, D., Menezes, A., Vanstone, S.: The elliptic curve digital signature
  algorithm ({ECDSA}). International Journal of Information Security
  \textbf{1}(1),  36--63 (2001). \doi{10.1007/s102070100002}

\bibitem{kasselman2026aiagentauth}
Kasselman, P., Lombardo, J., Rosomakho, Y., Campbell, B., Steele, N.: {AI}
  agent authentication and authorization. Internet-draft
  draft-klrc-aiagent-auth-01, IETF (2026), work in progress; authors from
  Defakto Security, AWS, Zscaler, Ping Identity, OpenAI

\bibitem{rfc5869}
Krawczyk, H., Eronen, P.: {HMAC}-based extract-and-expand key derivation
  function ({HKDF}). {RFC} 5869, IETF (2010)

\bibitem{menetrey2022attestation}
M\'{e}n\'{e}trey, J., G\"{o}ttel, C., Khurshid, A., Pasin, M., Felber, P.,
  Schiavoni, V., Raza, S.: Attestation mechanisms for trusted execution
  environments demystified. In: Proceedings of the 22nd International
  Conference on Distributed Applications and Interoperable Systems ({DAIS}).
  LNCS, vol. 13272, pp. 95--113. Springer (2022).
  \doi{10.1007/978-3-031-16092-9_7}

\bibitem{muid2025accurevoke}
Muid, M.R.A., Chung, T., Hoang, T.: {AccuRevoke}: Enhancing certificate
  revocation with distributed cryptographic accumulators. In: Proceedings of
  the 46th {IEEE} Symposium on Security and Privacy ({S\&P}). IEEE (2025)

\bibitem{raut2025txntokens}
Raut, A.: Transaction tokens for agents. Internet-draft
  draft-oauth-transaction-tokens-for-agents-00, IETF (2025), work in progress

\bibitem{rfc7662}
Richer, J.: {OAuth 2.0} token introspection. {RFC} 7662, IETF (2015)

\bibitem{rose2020nist}
Rose, S., Borchert, O., Mitchell, S., Connelly, S.: Zero trust architecture.
  Special publication 800-207, National Institute of Standards and Technology
  ({NIST}) (2020). \doi{10.6028/NIST.SP.800-207}

\bibitem{rfc6960}
Santesson, S., Myers, M., Ankney, R., Malpani, A., Galperin, S., Adams, C.:
  {X.509} internet public key infrastructure online certificate status protocol
  -- {OCSP}. {RFC} 6960, IETF (2013)

\bibitem{tobin2025delegation}
South, T., Marro, S., Hardjono, T., Mahari, R., Whitney, C., Chan, A.,
  Pentland, A.: Position: {AI} agents need authenticated delegation. In:
  Proceedings of the 42nd International Conference on Machine Learning ({ICML})
  (2025), position Paper Track (Oral)

\bibitem{w3c2025vcdm}
Sporny, M., Longley, D., Chadwick, D., Herman, I.: Verifiable credentials data
  model v2.0. {W3C} recommendation, World Wide Web Consortium ({W3C}) (2025),
  \url{https://www.w3.org/TR/vc-data-model-2.0/}

\bibitem{w3c2025bitstring}
Sporny, M., Longley, D., Prorock, M., Alkhraishi, M.: Bitstring status list
  v1.0. {W3C} recommendation, World Wide Web Consortium ({W3C}) (May 2025),
  \url{https://www.w3.org/TR/2025/REC-vc-bitstring-status-list-20250515/}

\bibitem{wuille2012bip32}
Wuille, P.: {BIP-32}: Hierarchical deterministic wallets. Bitcoin Improvement
  Proposal (2012),
  \url{https://github.com/bitcoin/bips/blob/master/bip-0032.mediawiki}

\end{thebibliography}

\end{document}